\documentclass[american,aps,prl,reprint,superscriptaddress]{revtex4-1}
\usepackage[T1]{fontenc}
\usepackage[latin9]{inputenc}
\usepackage{color}
\usepackage{babel}
\usepackage{amsmath}
\usepackage{amsthm}
\usepackage{amssymb}
\usepackage[unicode=true,pdfusetitle,
 bookmarks=true,bookmarksnumbered=false,bookmarksopen=false,
 breaklinks=false,pdfborder={0 0 0},backref=false,colorlinks=true,allcolors=magenta]
 {hyperref}
\usepackage{graphicx}
\usepackage{tikz-cd}

\newcommand{\Tr}{\mathrm{Tr}}
\newcommand{\ketbra}[1]{|#1\rangle\langle#1|}
\newcommand{\bra}[1]{\langle #1|}
\newcommand{\ket}[1]{|#1 \rangle}

\newcommand{\eps}{\varepsilon}
\newcommand*{\cL}{\mathcal{L}}

\newcommand*{\cH}{\mathcal{H}}
\newcommand*{\cF}{\mathcal{F}}
\newcommand*{\cD}{\mathcal{D}}
\newcommand*{\cG}{\mathcal{G}}
\newcommand*{\cK}{\mathcal{K}}
\newcommand*{\cE}{\mathcal{E}}
\newcommand*{\cP}{\mathcal{P}}

\newcommand{\suppress}[1]{}
\newcommand{\defeq}{\ensuremath{ \stackrel{\mathrm{def}}{=} }}
\newcommand{\F}{\mathrm{F}}
\newcommand{\Pur}{\mathrm{P}}
\newcommand {\br} [1] {\ensuremath{ \left( #1 \right) }}
\newcommand {\minusspace} {\: \! \!}
\newcommand {\smallspace} {\: \!}
\newcommand {\fn} [2] {\ensuremath{ #1 \minusspace \br{ #2 } }}
\newcommand {\ball} [2] {\fn{\mathcal{B}^{#1}}{#2}}
\newcommand {\relent} [2] {\fn{\mathrm{D}}{#1 \middle\| #2}}
\newcommand {\varrelent} [2] {\fn{\mathrm{V}}{#1 \middle\| #2}}
\newcommand {\dmax} [2] {\fn{\mathrm{D}_{\max}}{#1 \middle\| #2}}
\newcommand {\dmaxeps} [3] {\fn{\mathrm{D}^{#3}_{\max}}{#1 \middle\| #2}}
\newcommand {\mutinf} [2] {\fn{\mathrm{I}}{#1 \smallspace : \smallspace #2}}
\newcommand {\condmutinf} [3] {\mutinf{#1}{#2 \smallspace \middle\vert \smallspace #3}}
\newcommand {\dheps} [3] {\ensuremath{\mathrm{D}_{\mathrm{H}}^{#3}\left(#1 \| #2\right)}}
\newcommand {\dFeps} [3] {\ensuremath{\mathrm{fD}_{\mathrm{H}}^{#3}\left(#1 \| #2\right)}}
\newcommand {\id} {\ensuremath{\mathrm{I}}}

\newtheorem{definition}{Definition}
\newtheorem{claim}{Claim}
\newtheorem{corollary}{Corollary}
\newtheorem{fact}{Fact}
\newtheorem{theorem}{Theorem}
\newtheorem{lemma}{Lemma}

\makeatletter

\usepackage{times}
\usepackage{pxfonts}

\newcommand{\Co}{\mathcal{C}}
\newcommand{\Qu}{\mathcal{Q}}
\newcommand{\En}{\mathcal{E}}

\makeatother

\begin{document}

\title{Quantum state redistribution with local coherence}

\author{Anurag Anshu}

\affiliation{Centre for Quantum Technologies, National University of Singapore,
Singapore}

\author{Rahul Jain}

\affiliation{Center for Quantum Technologies, National University of Singapore
and MajuLab, UMI 3654, Singapore}

\author{Alexander Streltsov}

\affiliation{Faculty of Applied Physics and Mathematics, Gda\'{n}sk University
of Technology, 80-233 Gda\'{n}sk, Poland}

\affiliation{National Quantum Information Centre in Gda\'{n}sk, 81-824 Sopot,
Poland}
\begin{abstract}
Quantum entanglement and coherence are two fundamental resources for
quantum information processing. Recent results clearly demonstrate
their relevance in quantum technological tasks, including quantum
communication and quantum algorithms. In this Letter we study the
role of quantum coherence for quantum state redistribution, a fundamental
task where two parties aim to relocate a quantum particle by using
a limited amount of quantum communication and shared entanglement.
We provide general bounds for the resource rates required for this
process, and show that these bounds are tight under additional reasonable
constraints, including the situation where the receiving party cannot
use local coherence. While entanglement cannot be directly converted
into local coherence in our setting, we show that entanglement is
still useful for local coherence creation if an additional quantum
channel is provided, and the optimal protocol for local coherence
creation for any given amount of quantum communication and shared
entanglement is presented. We also discuss possible extensions of our methods to other scenarios where the receiving party is limited by local constraints, including theories of thermodynamics and asymmetry.
\end{abstract}
\maketitle
\textbf{\emph{Introduction}}\textbf{.} Recent advances in quantum
technology clearly demonstrate the power of quantum phenomena for
technological applications. An important example is quantum key distribution~\cite{Bennett84,EkertPhysRevLett.67.661},
which allows for provably secure communication between distant parties
by using quantum entanglement~\cite{EinsteinPhysRev.47.777,VedralPhysRevLett.78.2275,HorodeckiRevModPhys.81.865},
a type of correlations which does not exist in classical physics.
Recent experiments demonstrate first satellite-based quantum key distribution,
which uses entanglement established between Earth and a low-Earth-orbit
satellite~\cite{Liao2017,YinPhysRevLett.119.200501}. These developments
suggest that quantum technology will be part of our daily life within
the next years.

While quantum entanglement is one of the most important resources
for quantum technologies, several quantum technological applications
are not based on the presence of entanglement, but require other types
of nonclassicality. An important example is quantum coherence, which
arises from superposition principle of quantum mechanics. In the last
years a full resource theory of quantum coherence has been developed~\cite{BaumgratzPhysRevLett.113.140401,WinterPhysRevLett.116.120404,StreltsovRevModPhys.89.041003},
which allows for rigorous investigation of coherence for quantum technological
tasks, including such fundamental applications as quantum algorithms~\cite{HilleryPhysRevA.93.012111,Matera2058-9565-1-1-01LT01}
and quantum metrology~\cite{Giovannetti2011,MarvianPhysRevA.94.052324}. 

As any quantum resource theory, also the resource theories of coherence
and entanglement are based on the notion of \emph{free states} and
\emph{free operations}. Free states are quantum states which can be
prepared without additional cost. In entanglement theory these are
separable states, i.e., convex combinations of product states. In
coherence theory free states are quantum states which are diagonal
in a certain basis: the particular choice of the basis depends on
the specific problem under study, and is usually justified by the
unavoidable decoherence~\cite{ZurekRevModPhys.75.715}. Free operations
are transformations which can be performed without consumption of
resources. In entanglement theory this set is known as \emph{local
operations and classical communication}. In coherence theory this
set is usually identified with \emph{incoherent operations}: these
are quantum measurements which do not create coherence for individual
measurement outcomes \cite{BaumgratzPhysRevLett.113.140401}. Incoherent
operations admit an incoherent Kraus decomposition 
\begin{equation}
\Lambda_{\mathrm{IO}}[\rho]=\sum_{i}K_{i}\rho K_{i}^{\dagger},\label{eq:IO}
\end{equation}
where each of the Kraus operators $K_{i}$ cannot create coherence
individually. Other families of free operations for entanglement and
coherence theory have also been studied in the literature~\cite{HorodeckiRevModPhys.81.865,StreltsovRevModPhys.89.041003}. 

It has been recently recognized that both -- entanglement and coherence
-- are useful resources for quantum communication. Long-distance entanglement
is required to enable the parties to exchange quantum bits via teleportation~\cite{BennettPhysRevLett.70.1895,Portbased08,Portbased09,Ren2017}
in the first place. Additional use of local coherence can potentially
reduce the amount of entanglement in the protocol, when compared to
the situation where no local coherence is available. This effect has
been demonstrated for \emph{quantum state merging}~\cite{Horodecki2005,HorodeckiOW07,StreltsovCRBWL16},
where two parties -- Alice and Bob -- aim to merge their shares of
a quantum state on Bob's side while preserving correlations with the
environment. If Bob has no access to local coherence, the procedure
requires shared entanglement at rate~\cite{StreltsovCRBWL16} 
\begin{equation}
S(A|B)_{\overline{\rho}}=S(\overline{\rho}^{AB})-S(\overline{\rho}^{B}),
\end{equation}
which corresponds to the conditional entropy of the state $\overline{\rho}=\overline{\rho}^{AB}$.
Here, $S(\rho)=-\mathrm{Tr}[\rho\log_{2}\rho]$ is the von Neumann
entropy and 
\begin{equation}
\overline{\rho}^{X}=\sum_{i}\ket{i}\!\bra{i}^{X}\rho^{X}\ket{i}\!\bra{i}^{X}
\end{equation}
 denotes complete dephasing of a (possibly multipartite) system $X$
in the incoherent basis. Local coherence allows to reduce the consumption
of entanglement: in the optimal situation the required entanglement
rate corresponds to the conditional entropy of the state $\rho=\rho^{AB}$~\cite{Horodecki2005,HorodeckiOW07}
\begin{equation}
S(A|B)_{\rho}=S(\rho^{AB})-S(\rho^{B}).
\end{equation}
This result admits an operational interpretation even if the conditional
entropy is negative: in this case quantum state merging can be achieved
without entanglement, and additional singlets are gained at rate $-S(A|B)_{\rho}$.

In this Letter we study the role of shared entanglement and local
coherence for \emph{quantum state redistribution}, which is another
fundamental protocol for quantum communication~\cite{Devatakyard,YardD09}.
If not stated otherwise, we assume the asymptotic i.i.d. setting throughout this work. Before we present the main results, we recall the definition and aim
of quantum state redistribution in the following. 

\bigskip{}

\textbf{\emph{Quantum state redistribution.}} In the task of quantum
state redistribution Alice, Bob, and a referee share a total pure
state $\ket{\psi}^{RABC}$, where Alice initially holds the particles
$AC$ and Bob holds the particle $B$. The aim of the task is to send
the particle $C$ from Alice to Bob in such a way that the overall
state $\ket{\psi}$ remains intact. For achieving this, Alice and
Bob can use quantum communication at rate $\Qu$ and singlets at rate
$\En$. This task was first introduced in \cite{Devatakyard,YardD09},
and it has been shown that the process is possible if and only if
\begin{subequations}\label{eq:QE}
\begin{align}
\Qu & \geq\frac{1}{2}I(C\!:\!R|B),\\
\Qu+\En & \geq S(C|B),
\end{align}
\end{subequations}with the quantum conditional mutual information
\begin{equation}
I(C\!:\!R|B)=S(R|B)-S(R|CB).
\end{equation}
The result in Eqs.~(\ref{eq:QE}) can thus be seen as the first operational
interpretation of the quantum conditional mutual information~\cite{Devatakyard,YardD09}. 

Note that in contrast to many other quantum communication tasks --
such as quantum state merging~\cite{Horodecki2005,HorodeckiOW07}
-- classical communication is not free in the setting considered here.
In general, quantum communication is more powerful than shared entanglement:
Alice and Bob can establish one singlet between each other by sending
one qubit, but they cannot use singlets (without classical communication)
for exchanging qubits. Special cases of interest are the \emph{fully
quantum Slepian-Wolf task} and \emph{quantum state splitting}~\cite{AbeyesingheDHW09}.
The former arises if Alice lacks side information, i.e., for total
states of the form $\ket{\psi}^{RBC}\otimes\ket{\phi}^{A}$. Correspondingly,
quantum state splitting arises if Bob has no side information,
i.e., for total states of the form $\ket{\psi}^{RAC}\otimes\ket{\phi}^{B}$.

\bigskip{}

\textbf{\emph{Quantum state redistribution and local coherence. }}In
standard quantum state redistribution \cite{Devatakyard,YardD09}
both parties, Alice and Bob, have access to local quantum coherence
at no cost, i.e., they can perform arbitrary quantum operations locally.
This is the main difference to the scenario considered in this Letter:
here we study the situation where Bob has restricted access to local
coherence. In more detail, we assume that Alice has access to arbitrary
quantum operations locally, while Bob is restricted to local incoherent
operations, see Eq.~(\ref{eq:IO}) for their definition. Additionally,
Bob is provided with local maximally coherent qubits 
\begin{equation}
\ket{+}=\frac{1}{\sqrt{2}}(\ket{0}+\ket{1})
\end{equation}
 at rate $\Co$, thus potentially allowing Bob to perform more general
operations locally. In fact, Bob can implement an arbitrary $n$-qubit
quantum operation locally if he has access to local incoherent operations
and $n$ additional maximally coherent qubits~\cite{BenDanaPhysRevA.95.062327,ChitambarPhysRevLett.117.020402}.
We consider resource triples $(\Qu,\En,\Co)$, where $\Co$ denotes
the rate of local coherence at Bob's side.

Having introduced the general setting, it is instrumental to compare
it to the framework of \emph{local quantum-incoherent operations and
classical communication (LQICC)}~\cite{ChitambarPhysRevLett.116.070402,StreltsovPhysRevX.7.011024}.
In this framework, Alice can perform arbitrary quantum operations
locally, while Bob is restricted to local incoherent operations, and
both parties can exchange their measurement outcomes via classical
communication, which is provided for free in this setting. In the
LQICC scenario entanglement is a more powerful resource than coherence,
as via LQICC operations one singlet can be converted into one maximally
coherent qubit on Bob's side~\cite{ChitambarPhysRevLett.116.070402}.
Without classical communication the situation changes significantly:
it is not possible to convert singlets into maximally coherent states
on Bob's side if no further communication is provided. However, as
we will see in Theorem~\ref{thm:1-1} below, singlets are still useful
for local coherence extraction in the presence of an additional quantum
channel. 

In the following we will make use of the relative entropy of coherence~\cite{BaumgratzPhysRevLett.113.140401}
\begin{equation}
R_{\mathrm{c}}(\rho)=\min_{\sigma\in\mathcal{I}}S(\rho||\sigma),
\end{equation}
where $\mathcal{I}$ denotes the set of incoherent states. The relative
entropy of coherence admits the closed expression $R_{\mathrm{c}}(\rho)=S(\overline{\rho})-S(\rho)$
and coincides with the distillable coherence of the state $\rho$,
i.e., the optimal rate for asymptotic extraction of maximally coherent
qubits via incoherent operations~\cite{WinterPhysRevLett.116.120404}.

Equipped with these tools we are now in position to prove the following
Theorem. Here, $q$, $e$, and $c$ denote absolute numbers of exchanged
qubits, shared singlets, and local maximally coherent qubits on Bob's
side~\footnote{This is in contrast to $\Qu$, $\En$, and $\Co$, which denote \emph{rates}
of the corresponding resources.}. 
\begin{theorem}
\label{thm:1-1}If Alice and Bob share $e$ singlets, by sending $q$
additional qubits from Alice to Bob it is possible to establish 
\begin{equation}
c=q+\min\{e,q\}\label{eq:c}
\end{equation}
maximally coherent qubits on Bob's side. For fixed $q$ and $e$ this
amount of coherence is maximal.\end{theorem}
\begin{proof}
We first study the case $e\leq q$. We will now prove that there exists
a protocol for establishing $q+e$ maximally coherent qubits on Bob's
side. For this, we recall that Alice can perform arbitrary quantum
operations locally. Since all maximally entangled states are related
via local unitaries on one party only, this means that Alice can bring
all shared singlets into the form 
\begin{equation}
\ket{\psi}^{AB}=\frac{1}{\sqrt{2}}(\ket{0}^{A}\ket{+}^{B}+\ket{1}^{A}\ket{-}^{B})\label{eq:Singlet}
\end{equation}
with maximally coherent states $\ket{\pm}=(\ket{0}\pm\ket{1})/\sqrt{2}$.
Alice can now send her parts of all the singlets to Bob, which requires
$e$ qubits of quantum communication. Bob -- who is now holding $e$
copies of the state $\ket{\psi}$ -- can convert each copy into the
state $\ket{+}\ket{+}$ by applying an incoherent unitary locally.
This proves that by using $e$ shared singlets and sending $e$ qubits
from Alice to Bob, it is possible to establish $2e$ maximally coherent
qubits at Bob's side. Alice now uses the remaining $q-e$ qubits of
quantum communication for sending $q-e$ maximally coherent qubits
to Bob. The final number of maximally coherent qubits that Bob obtains
in this way is $q+e$, as claimed.

We will now show that the protocol presented above is optimal, i.e.,
it is not possible to establish more than $q+e$ maximally coherent
qubits on Bob's side. This can be proven by contradiction, assuming
that some protocols achieves $c>q+e$ by using $e$ shared singlets
and sending $q$ qubits from Alice to Bob. It is now crucial to note
that teleportation of one qubit from Alice to Bob is possible by using
LQICC operations and one shared singlet~\cite{StreltsovPhysRevX.7.011024}.
Thus, the above assumption leads to the conclusion that $e'=q+e$
singlets together with LQICC operations can be used to establish $c>e'$
maximally coherent qubits at Bob's side, which is a contradiction
to the results presented in~\cite{ChitambarPhysRevLett.116.070402}.

To complete the proof, we will consider the case $e>q$. The protocol
for establishing $2q$ maximally coherent states at Bob's side is
the same as above, i.e., Alice uses $q$ qubits to send her parts
of $q$ maximally entangled states to Bob, who can then locally convert
each copy into the state $\ket{+}\ket{+}$ via local incoherent unitary.
It remains to show that this protocol is optimal. This is a consequence
of Lemma~\ref{lem:4} in Supplemental Material, which
shows that Bob's local relative entropy of coherence $R_{\mathrm{c}}$
cannot grow more than $2\log_{2}d$ if a particle of dimension $d$
is sent from Alice to Bob. Thus, by sending $q$ qubits from Alice
to Bob, Bob's local relative entropy of coherence cannot grow more
than $2q$, which implies that $c\leq2q$. This completes the proof
of the Theorem. 
\end{proof}
We will now go back to the main problem of quantum state redistribution
and discuss several important cases, starting with the fully quantum
Slepian-Wolf task. Here, Alice, Bob, and the referee initially
share a total state of the form $\ket{\psi}^{RBC}\otimes\ket{\phi}^{A}$,
and the particle $C$ is sent from Alice to Bob. The theorem below
provides a bound on the resource triple $(\Qu,\En,\Co)$ required
for performing state redistribution in this setting. 
\begin{theorem}
\label{thm:1}A necessary condition for achieving quantum state redistribution
from Alice to Bob for a state $\ket{\psi}^{RBC}\otimes\ket{\phi}^{A}$
is that 
\begin{equation}
\Qu+\En+\Co\geq S(\overline{\rho}^{BC})-S(\overline{\rho}^{B}).\label{eq:Spelian-Wolf}
\end{equation}
\end{theorem}
\begin{proof}
As was shown in~\cite{StreltsovCRBWL16}, the following condition
is necessary for quantum state merging of the state $\ket{\psi}^{RBC}\otimes\ket{\phi}^{A}$
via LQICC operations: 
\begin{equation}
\En+\Co\geq S(\overline{\rho}^{BC})-S(\overline{\rho}^{B}).\label{eq:Merging}
\end{equation}
Assume now -- by contradiction -- that there is a quantum state redistribution
protocol operating with a resource tripe $(\Qu,\En,\Co)$ violating
Eq.~(\ref{eq:Spelian-Wolf}). Then, due to results presented in~\cite{StreltsovCRBWL16,StreltsovPhysRevX.7.011024},
there exists an LQICC protocol which achieves quantum state merging
by using entanglement at rate $\En'=\Qu+\En$ and local coherence
at Bob's side at rate $\Co'=\Co$. The proof is complete by noting
that the resource pair $(\En',\Co')$ violates Eq.~(\ref{eq:Merging}).
\end{proof}
The result in Theorem~\ref{thm:1} provides a general lower bound
on the resource rates required for quantum state redistribution without
side information on Alice's side. We will now go one step further
by giving the optimal quantum communication rate $\Qu$ required for
quantum state redistribution under the assumption that Alice and Bob
only use forward quantum communication, and Bob applies local incoherent
operations without any local coherence, i.e., $\Co=0$. 
\begin{theorem}
\label{thm:Q}Quantum state redistribution with forward quantum communication
at rate $\Qu$ and without local coherence on Bob's side is possible
if and only if 
\begin{equation}
\Qu\geq\frac{1}{2}\left\{ I(C\!:\!R|B)+R_{\mathrm{c}}(\rho^{BC})-R_{\mathrm{c}}(\rho^{B})\right\} .\label{eq:Qinc}
\end{equation}

\end{theorem}
\noindent We refer to Supplemental Material for the
proof. After the particle $C$ is transferred, the local resource
(coherence) that Bob has is $R_{\mathrm{c}}(\rho^{BC})$ whereas the
resource he started with is $R_{\mathrm{c}}(\rho^{B})$. Hence, the
process requires additional $[R_{\mathrm{c}}(\rho^{BC})-R_{\mathrm{c}}(\rho^{B})]/2$
qubits of communication, on top of $I(C\!:\!R|B)/2$ qubits of communication
which is needed if Bob has no local restrictions.

We will now demonstrate the power of the above results on specific
examples. If both, Alice and Bob, lack side information, Theorem~\ref{thm:Q}
reduces to 
\begin{equation}
\Qu\geq\frac{1}{2}\left\{ S(\rho^{C})+S(\overline{\rho}^{C})\right\} .\label{eq:Schumacher}
\end{equation}
This scenario can be regarded as an incoherent version of Schumacher
compression~\cite{SchumacherPhysRevA.51.2738}, where the decompression
stage is restricted to incoherent operations. The result in Eq.~(\ref{eq:Schumacher})
should be compared to the standard Schumacher compression rate $S(\rho^{C})$~\cite{SchumacherPhysRevA.51.2738},
and to the minimal singlet rate for state merging via LQICC operations
in this setting, which is given by $S(\overline{\rho}^{C})$~\cite{StreltsovCRBWL16}.
Remarkably, the right-hand side of Eq.~(\ref{eq:Schumacher}) is
in general below the optimal LQICC singlet rate $S(\overline{\rho})$.
This is due to the fact that apart from quantum communication at rate
$\{S(\rho^{C})+S(\overline{\rho}^{C})\}/2$ state redistribution also
consumes additional singlets. 

To make this comparison more explicit, consider the situation where
the transmitted particle $C$ is a maximally coherent qubit $\rho^{C}=\ket{+}\!\bra{+}^{C}$,
not correlated with Alice or Bob. Due to Eq.~(\ref{eq:Schumacher}),
it follows that -- in the presence of additional singlets -- Alice
can send this state to Bob by using quantum communication at rate
$\Qu=1/2$. This is indeed achieved by using the same method as in
the proof of Theorem~\ref{thm:1-1}: Alice first brings all singlets
into the form~(\ref{eq:Singlet}) via suitable local unitaries, and
then sends her half of the singlets to Bob. After performing an incoherent
unitary locally, Bob obtains two maximally coherent qubits for each
qubit of quantum communication. However, if Alice and Bob are restricted
to LQICC operations, they will need singlets at rate $\En=1$ for
sending the state $\rho^{C}=\ket{+}\!\bra{+}^{C}$~\cite{ChitambarPhysRevLett.116.070402}.

For the fully quantum Slepian-Wolf task Theorem~\ref{thm:Q} reduces
as follows:
\begin{equation}
\Qu\geq\frac{1}{2}\left\{ I(C\!:\!R)+R_{\mathrm{c}}(\rho^{BC})-R_{\mathrm{c}}(\rho^{B})\right\} .\label{eq:Slepian-Wolf}
\end{equation}
Correspondingly, for quantum state splitting we obtain
\begin{equation}
\Qu \geq \frac{1}{2}\left\{ I(C\!:\!R)+R_{\mathrm{c}}(\rho^{C})\right\}. \label{eq:splitting}
\end{equation}
We will discuss these results in more detail in the next section.

\bigskip{}

\textbf{\emph{Quantum state redistribution with entanglement-assisted
classical communication.}} A fundamental result in quantum information
theory is the fact that one qubit can transfer two bits of classical
information by consuming an additional singlet, a phenomenon also
known as \emph{superdense coding}~\cite{bennett92}. On the other
hand, two bits of classical communication together with one singlet
can be used to teleport one qubit~\cite{BennettPhysRevLett.70.1895},
see also \cite{Jain2012} and references therein.

Note that both -- superdense coding and teleportation -- can be performed
by using only incoherent operations on the recipient's side. Thus,
Theorem~\ref{thm:Q} implies that $I\left(C\!:\!R|B\right)+R_{\mathrm{c}}(\rho^{BC})-R_{\mathrm{c}}(\rho^{B})$
is the minimal forward classical communication rate in the presence
of additional singlets and in the absence of coherence on Bob's side.

Together with Eq.~(\ref{eq:Slepian-Wolf}), this implies that the
fully quantum Slepian-Wolf task requires forward classical communication
at rate $I(C\!:\!R)+R_{\mathrm{c}}(\rho^{BC})-R_{\mathrm{c}}(\rho^{B})$,
if Bob is restricted to local incoherent operations locally. This
should be compared to the classical communication rate for standard
quantum state merging, given by the mutual information $I(C\!:\!R)$~\cite{Horodecki2005,HorodeckiOW07}.
Thus, the increase of the classical communication rate due to restrictions
on Bob's side is given by the increase of Bob's local coherence: $R_{\mathrm{c}}(\rho^{BC})-R_{\mathrm{c}}(\rho^{B})$. Correspondingly, due to Eq.~(\ref{eq:splitting}), the forward classical communication rate for quantum state splitting is given by $I(C\!:\!R)+R_{\mathrm{c}}(\rho^{C})$, again under the assumption that Bob uses only incoherent operations locally and that additional singlets are available. 

These results can be extended to other scenarios, where Bob's local operations are not necessarily incoherent. In particular, if Bob can swap local particles for free, then there is an optimal protocol for quantum state splitting which uses only swap operations on Bob's side, again under the constraint that Bob has no further local resource available and that shared entanglement is provided~\footnote{Here, "optimal" means that the protocol has minimal classical communication rate. The protocol can nevertheless use high amount of entanglement.}. Apart from the resource theory of coherence discussed in this Letter, this result covers other settings where Bob has local restrictions in his lab, including the resource theories of thermodynamics~\cite{BrandaoPhysRevLett.111.250404,LostaglioPhysRevX.5.021001}, purity~\cite{HorodeckiPhysRevA.67.062104,Gour20151}, and asymmetry~\cite{Gour1367-2630-10-3-033023,GourPhysRevA.80.012307}. We refer to the Supplemental Material for the proof and more details.

\bigskip{}

\textbf{\emph{Conclusions. }}In this Letter we have studied the role
of local coherence for quantum state redistribution, assuming in particular
that the receiving party can only use local incoherent operations
for free. We have studied the interplay of the quantum communication
rate $\Qu$, the singlet rate $\En$, and the rate of local coherence
$\Co$ in this scenario, and proved several important results. 

In the absence of quantum communication, entanglement cannot be converted
into local coherence in this setting. However, entanglement is still
helpful in this procedure if a quantum channel is available, and we
have presented the optimal protocol for local coherence creation for
any given amount of quantum communication and shared entanglement. 

If Alice has no side information, we showed that the sum of all the
rates $\Qu+\En+\Co$ in quantum state redistribution is bounded below
by the conditional entropy of $\overline{\rho}$, the latter being
the state of Alice and Bob after applying a complete dephasing in
the incoherent basis. Our results further lead to the optimal quantum
communication rate if Bob is not using local coherence. Counterintuitively,
we showed that in some situations the distribution protocol requires
less quantum communication, when compared to the required amount of
entanglement in the presence of free classical communication. We also
discussed important applications of these results, including the incoherent
versions of fully quantum Slepian-Wolf task, quantum state splitting, and Schumacher compression. Some of our results also apply to more general settings where Bob has other local restrictions, not necessary arising from coherence theory.

\bigskip{}

\textbf{\emph{Acknowledgements. }}We thank Andreas Winter for discussion.
A.A. and R.J. are supported by the Singapore Ministry of Education
and the National Research Foundation, through the Tier 3 Grant \textquotedblleft Random
numbers from quantum processes\textquotedblright{} MOE2012-T3-1-009.
A.S. acknowledges financial support by the National Science Center
in Poland (POLONEZ UMO-2016/21/P/ST2/04054) and the European Union\textquoteright s
Horizon 2020 research and innovation programme under the Marie Sk\l odowska-Curie
grant agreement No. 665778.

\bibliographystyle{naturemag}
\bibliography{References}

\clearpage
\onecolumngrid

\appendix

\section*{Supplemental Material}

\section{Preliminaries}
\label{sec:prelims}

Consider a finite dimensional Hilbert space $\cH$ endowed with an inner product $\langle \cdot, \cdot \rangle$ (In this paper, we only consider finite dimensional Hilbert-spaces). The $\ell_1$ norm of an operator $X$ on $\cH$ is $\| X\|_1:=\Tr\sqrt{X^{\dagger}X}$ and $\ell_2$ norm is $\| X\|_2:=\sqrt{\Tr XX^{\dagger}}$. A quantum state (or a density matrix or a state) is a positive semi-definite matrix on $\cH$ with trace equal to $1$. It is called {\em pure} if and only if its rank is $1$. A sub-normalized state is a positive semi-definite matrix on $\cH$ with trace less than or equal to $1$. Let $\ket{\psi}$ be a unit vector on $\cH$, that is $\langle \psi,\psi \rangle=1$.  With some abuse of notation, we use $\psi$ to represent the state and also the density matrix $\ketbra{\psi}$, associated with $\ket{\psi}$. Given a quantum state $\rho$ on $\cH$, {\em support of $\rho$}, called $\text{supp}(\rho)$ is the subspace of $\cH$ spanned by all eigen-vectors of $\rho$ with non-zero eigenvalues.
 
A {\em quantum register} $A$ is associated with some Hilbert space $\cH_A$. Define $|A| := \dim(\cH_A)$. Let $\mathcal{L}(A)$ represent the set of all linear operators on $\cH_A$. We denote by $\mathcal{D}(A)$, the set of quantum states on the Hilbert space $\cH_A$. The set of subnormalized states are represented by $\cP(A)$. State $\rho$ with subscript $A$ indicates $\rho_A \in \mathcal{D}(A)$. If two registers $A,B$ are associated with the same Hilbert space, we shall represent the relation by $A\equiv B$.  Composition of two registers $A$ and $B$, denoted $AB$, is associated with Hilbert space $\cH_A \otimes \cH_B$.  For two quantum states $\rho\in \mathcal{D}(A)$ and $\sigma\in \mathcal{D}(B)$, $\rho\otimes\sigma \in \mathcal{D}(AB)$ represents the tensor product (Kronecker product) of $\rho$ and $\sigma$. The identity operator on $\cH_A$ (and associated register $A$) is denoted $\id_A$. 

Let $\rho_{AB} \in \mathcal{D}(AB)$. We define
\[ \rho_{B} := \Tr_{A}{\rho_{AB}}
:= \sum_i (\bra{i} \otimes \id_{B})
\rho_{AB} (\ket{i} \otimes \id_{B}) , \]
where $\{\ket{i}\}_i$ is an orthonormal basis for the Hilbert space $\cH_A$.
The state $\rho_B\in \mathcal{D}(B)$ is referred to as the marginal state of $\rho_{AB}$. Unless otherwise stated, a missing register from subscript in a state will represent partial trace over that register. Given a $\rho_A\in\mathcal{D}(A)$, a {\em purification} of $\rho_A$ is a pure state $\rho_{AB}\in \mathcal{D}(AB)$ such that $\Tr_{B}{\rho_{AB}}=\rho_A$. Purification of a quantum state is not unique. Given two registers $A$ and $B$, $SEP(A:B)$ denotes the set of all separable states across $A, B$, that is, the set of all states $\rho_{AB}\in \cD(AB)$ such that $\rho_{AB}=\sum_k p_k \rho^k_A\otimes \rho^k_B$, where $\rho^k_A\in \cD(A), \rho^k_B \in \cD(B)$ and $\sum_k p_k=1$.

A quantum {map} $\cE: \mathcal{L}(A)\rightarrow \mathcal{L}(B)$ is a completely positive and trace preserving (CPTP) linear map (mapping states in $\mathcal{D}(A)$ to states in $\mathcal{D}(B)$). A {\em unitary} operator $U_A:\cH_A \rightarrow \cH_A$ is such that $U_A^{\dagger}U_A = U_A U_A^{\dagger} = \id_A$. An {\em isometry}  $V:\cH_A \rightarrow \cH_B$ is such that $V^{\dagger}V = \id_A$ and $VV^{\dagger} = \id_B$. The set of all unitary operations on register $A$ is  denoted by $\mathcal{U}(A)$.

\begin{definition}
We shall consider the following information theoretic quantities. Reader is referred to ~\cite{Renner05, Tomamichel09,Tomamichel12,Datta09} for many of these definitions. Let $\varepsilon \in (0,1)$. 
\begin{enumerate}
\item {\bf Fidelity} For $\rho_A,\sigma_A \in \mathcal{D}(A)$, $$\F(\rho_A,\sigma_A)\defeq\|\sqrt{\rho_A}\sqrt{\sigma_A}\|_1.$$ For classical probability distributions $P = \{p_i\}, Q =\{q_i\}$, $$\F(P,Q)\defeq \sum_i \sqrt{p_i \cdot q_i}.$$
\item {\bf Purified distance} For $\rho_A,\sigma_A \in \mathcal{D}(A)$, $$\Pur(\rho_A,\sigma_A) = \sqrt{1-\F^2(\rho_A,\sigma_A)}.$$
\item {\bf $\varepsilon$-ball} For $\rho_A\in \mathcal{D}(A)$, $$\ball{\eps}{\rho_A} \defeq \{\rho'_A\in \mathcal{D}(A)|~\Pur(\rho_A,\rho'_A) \leq \varepsilon\}. $$ 

\item {\bf Von-neumann entropy} For $\rho_A\in\mathcal{D}(A)$, $$S(\rho_A) \defeq - \Tr(\rho_A\log\rho_A) .$$ 
\item {\bf Relative entropy} For $\rho_A,\sigma_A\in \mathcal{D}(A)$ such that $\text{supp}(\rho_A) \subset \text{supp}(\sigma_A)$, $$\relent{\rho_A}{\sigma_A} \defeq \Tr(\rho_A\log\rho_A) - \Tr(\rho_A\log\sigma_A) .$$ 
\item {\bf Max-relative entropy} For $\rho_A,\sigma_A\in \mathcal{D}(A)$ such that $\text{supp}(\rho_A) \subset \text{supp}(\sigma_A)$, $$ \dmax{\rho_A}{\sigma_A}  \defeq  \inf \{ \lambda \in \mathbb{R} : 2^{\lambda} \sigma_A \succeq \rho_A \}  .$$  
\item {\bf Hypothesis testing relative entropy}  For $\rho_A,\sigma_A\in \mathcal{D}(A)$, $$ \dheps{\rho_A}{\sigma_A}{\eps}  \defeq  \sup_{0\preceq \Pi\preceq \id, \Tr(\Pi\rho_A)\geq 1-\eps}\log\left(\frac{1}{\Tr(\Pi\sigma_A)}\right).$$  
\item {\bf Restricted hypothesis testing relative entropy}  For $\rho_A,\sigma_A\in \cP(A)$, $$ \dFeps{\rho_A}{\sigma_A}{\eps}  \defeq  \sup_{0\preceq \Pi\preceq \id, \Pi \in \cF_{\cE}, \Tr(\Pi\rho_A)\geq 1-\eps}\log\left(\frac{1}{\Tr(\Pi\sigma_A)}\right).$$ If $\Tr(\rho_A) < 1-\eps$, then we set $\Pi=\id$. 
\item {\bf Mutual information} For $\rho_{AB}\in \mathcal{D}(AB)$, $$\mutinf{A}{B}_{\rho}\defeq S(\rho_A) + S(\rho_B)-S(\rho_{AB}) = \relent{\rho_{AB}}{\rho_A\otimes\rho_B}.$$
\item {\bf Conditional mutual information} For $\rho_{ABC}\in\mathcal{D}(ABC)$, $$\condmutinf{A}{B}{C}_{\rho}\defeq \mutinf{A}{BC}_{\rho}-\mutinf{A}{C}_{\rho}.$$
\end{enumerate}
\label{def:infquant}
\end{definition}	

We will use the following facts. 
\begin{fact}[Triangle inequality for purified distance,~\cite{Tomamichel12}]
\label{fact:trianglepurified}
For states $\rho_A, \sigma_A, \tau_A\in \mathcal{D}(A)$,
$$\Pur(\rho_A,\sigma_A) \leq \Pur(\rho_A,\tau_A)  + \Pur(\tau_A,\sigma_A) . $$ 
\end{fact}

\begin{fact}[\cite{stinespring55}](\textbf{Stinespring representation})\label{stinespring}
Let $\cE(\cdot): \mathcal{L}(A)\rightarrow \mathcal{L}(B)$ be a quantum operation. There exists a register $C$ and an unitary $U\in \mathcal{U}(ABC)$ such that $\cE(\omega)=\Tr_{A,C}\br{U (\omega  \otimes \ketbra{0}^{B,C}) U^{\dagger}}$. Stinespring representation for a channel is not unique. 
\end{fact}

\begin{fact}[Monotonicity under quantum operations, \cite{barnum96},\cite{lindblad75}]
	\label{fact:monotonequantumoperation}
For quantum states $\rho$, $\sigma \in \mathcal{D}(A)$, and quantum operation $\cE(\cdot):\mathcal{L}(A)\rightarrow \mathcal{L}(B)$, it holds that
\begin{align*}
	\|\cE(\rho) - \cE(\sigma)\|_1 \leq \|\rho - \sigma\|_1 \quad \mbox{and} \quad \F(\cE(\rho),\cE(\sigma)) \geq \F(\rho,\sigma) \quad \mbox{and} \quad \relent{\rho}{\sigma}\geq \relent{\cE(\rho)}{\cE(\sigma)}.
\end{align*}
In particular, for bipartite states $\rho_{AB},\sigma_{AB}\in \mathcal{D}(AB)$, it holds that
\begin{align*}
	\|\rho_{AB} - \sigma_{AB}\|_1 \geq \|\rho_A - \sigma_A\|_1 \quad \mbox{and} \quad \F(\rho_{AB},\sigma_{AB}) \leq \F(\rho_A,\sigma_A) \quad \mbox{and} \quad \relent{\rho_{AB}}{\sigma_{AB}}\geq \relent{\rho_A}{\sigma_A} .
\end{align*}
\end{fact}

\begin{fact}[Uhlmann's Theorem \cite{uhlmann76}]
\label{uhlmann}
Let $\rho_A,\sigma_A\in \mathcal{D}(A)$. Let $\rho_{AB}\in \mathcal{D}(AB)$ be a purification of $\rho_A$ and $\sigma_{AC}\in\mathcal{D}(AC)$ be a purification of $\sigma_A$. There exists an isometry $V: \cH_C \rightarrow \cH_B$ such that,
 $$\F(\ketbra{\theta}_{AB}, \ketbra{\rho}_{AB}) = \F(\rho_A,\sigma_A) ,$$
 where $\ket{\theta}_{AB} = (\id_A \otimes V) \ket{\sigma}_{AC}$.
\end{fact}

\begin{fact}[Gentle measurement lemma \cite{Winter:1999,Ogawa:2002}]
\label{gentlelemma}
Let $\rho$ be a quantum state and $0\preceq A\preceq I$ be an operator. Then 
$$\F(\rho, \frac{A\rho A}{\Tr(A^2\rho)})\geq \sqrt{\Tr(A^2\rho)}.$$
\end{fact}

\begin{fact} [Neumark's Theorem] \label{Neumark}For any POVM $\left\{A_i\right\}_{i \in \mathcal{I}}$ acting on a system $S,$ there exists a unitary $U_{SP}$ and an orthonormal basis $\left\{\ket{i}_P\right\}_{i \in \mathcal{I}}$ such that for all quantum states $\rho_S$, we have
$$\Tr_P\left[U^\dagger_{SP} \left(\mathbb{I}_S \otimes \ket{i}\bra{i}_P\right)U_{SP}\left(\rho_S \otimes \ket{0}\bra{0}_P\right)\right] = A_i \rho_SA_i^{\dagger}.$$
\end{fact}

\begin{fact}[Convex split lemma \cite{AnshuDJ14, AnshuJW17SR}]
\label{convexcombcor}
Let $\rho_{PQ}\in\mathcal{D}(PQ)$ and $\sigma_Q\in\mathcal{D}(Q)$ be quantum states such that $\text{supp}(\rho_Q)\subset\text{supp}(\sigma_Q)$.  Let $k \defeq \inf_{\rho'_{PQ}\in\ball{\eps}{\rho_{PQ}}}\dmax{\rho'_{PQ}}{\rho'_P\otimes\sigma_Q}$. Define the following state
\begin{equation*}
\tau_{PQ_1Q_2\ldots Q_n} \defeq  \frac{1}{n}\sum_{j=1}^n \rho_{PQ_j}\otimes\sigma_{Q_1}\otimes \sigma_{Q_2}\ldots\otimes\sigma_{Q_{j-1}}\otimes\sigma_{Q_{j+1}}\ldots\otimes\sigma_{Q_n}
\end{equation*}
on $n+1$ registers $P,Q_1,Q_2,\ldots Q_n$, where $\forall j \in [n]: \rho_{PQ_j} = \rho_{PQ}$ and $\sigma_{Q_j}=\sigma_Q$.  For $\delta \in (0,1) $ and $ n= \lceil\frac{2^k}{\delta}\rceil$, it holds that 
$$ \F^2(\tau_{PQ_1Q_2\ldots Q_n},\tau_P \otimes \sigma_{Q_1}\otimes\sigma_{Q_2}\ldots \otimes \sigma_{Q_n}) \geq 1-  (\sqrt{\delta}+2\eps)^2.$$ 
\end{fact}

\begin{fact}[\cite{TomHay13, li2014}]
\label{dmaxequi}
Let $\eps\in (0,1)$ and $n$ be an integer. Let $\rho^{\otimes n}, \sigma^{\otimes n}$ be quantum states. Define $\Phi(x) = \int_{-\infty}^x \frac{e^{-t^2/2}}{\sqrt{2\pi}} dt$. It holds that
\begin{equation*}
\dmaxeps{\rho^{\otimes n}}{\sigma^{\otimes n}}{\eps} = n\relent{\rho}{\sigma} + \sqrt{n\varrelent{\rho}{\sigma}} \Phi^{-1}(\eps) + O(\log n) ,
\end{equation*}
and 
\begin{equation*}
\dheps{\rho^{\otimes n}}{\sigma^{\otimes n}}{\eps} = n\relent{\rho}{\sigma} + \sqrt{n\varrelent{\rho}{\sigma}} \Phi^{-1}(\eps) + O(\log n) .
\end{equation*}
\end{fact}

\begin{fact}
\label{gaussianupper}
For the function $\Phi(x) = \int_{-\infty}^x \frac{e^{-t^2/2}}{\sqrt{2\pi}} dt$ and $\eps\leq \frac{1}{2}$, it holds that $|\Phi^{-1}(\eps)| \leq 2\sqrt{\log\frac{1}{2\eps}}$.
\end{fact}
\begin{proof}
We have $$\Phi(-x)=\int_{-\infty}^{-x} \frac{e^{-t^2/2}}{\sqrt{2\pi}} dt = \int_{0}^{\infty} \frac{e^{-(-x-t)^2/2}}{\sqrt{2\pi}} dt \leq e^{-x^2/2} \int_{0}^{\infty} \frac{e^{-(-t)^2/2}}{\sqrt{2\pi}} dt = \frac{1}{2}e^{-x^2/2}.$$ Thus, $\Phi^{-1}(\eps) \geq -2\sqrt{\log\frac{1}{2\eps}}$, which completes the proof.
\end{proof}

\begin{fact} [\cite{AnshuJW17}]
\label{closestatesmeasurement}
Let $\rho,\sigma$ be quantum states such that $\Pur(\rho,\sigma)\leq \eps$. Let $0\leq \Pi\leq \id$ be an operator such that $\Tr(\Pi\rho)\geq 1-\delta^2$. Then $\Tr(\Pi\sigma)\geq 1- (2\eps+\delta)^2$. If $\delta=0$, then $\Tr(\Pi\sigma) \geq 1-\eps^2$.
\end{fact}

\begin{fact}[\cite{Sen12, Gao15}]
\label{noncommutativebound}
Let $\rho$ be a quantum state and $\Pi_1, \Pi_2, \ldots \Pi_k$ be projectors. Let $\Pi'_i \defeq \id - \Pi_i$. Then 
$$\Pur\left(\frac{\Pi'_k\ldots\Pi'_2\Pi'_1\rho\Pi'_1\Pi'_2\ldots\Pi'_k}{\Tr\left(\Pi'_k\ldots\Pi'_2\Pi'_1\rho\Pi'_1\Pi'_2\ldots\Pi'_k\right)}, \rho\right)\leq  \left(\sum_i\Tr(\Pi_i\rho)\right)^{1/4}.$$
\end{fact}

\begin{fact}[Fannes inequality \cite{fannes73}]
\label{fact:fannes}
Given quantum states $\rho_1,\rho_2\in \cD(\cH_A)$, such that $|A|=d$ and $\Pur(\rho_1,\rho_2)= \eps \leq \frac{1}{2\mathrm{e}}$, $$|S(\rho_1)-S(\rho_2)|\leq \eps\log(d)+1.$$   
\end{fact}

\begin{fact}[\cite{DHorodecki99}]
\label{relentresourcecont}
Let $\rho,\rho' \in \mathcal{D}(\cH_M)$ be the quantum states on register $M$ with $\|\rho-\rho'\|_1:= \eps \leq \frac{1}{3}$. Let $\cF\subseteq \mathcal{D}(\cH_M)$ be a convex set. Then it holds that 
$$|\inf_{\sigma\in \cF}\relent{\rho}{\sigma} - \inf_{\sigma'\in \cF}\relent{\rho'}{\sigma'}| \leq \eps \left(\log M+ \inf_{\tau\in \cF}\|\log\tau\|_\infty\right) + \eps\log\frac{1}{\eps} + 4\eps.$$
\end{fact}

We have the following lemma for the resource theory of coherence.

\begin{lemma}
\label{lem:4}
For any quantum state $\rho^{AB}\in \cD(\cH_{AB})$ the following inequality
holds:
\begin{equation}
R_{\mathrm{c}}(\rho^{AB})-R_{\mathrm{c}}(\rho^{B})\leq2\log_{2}|A|,
\label{eq:lemma}
\end{equation}
where $R_{\mathrm{c}}$ is the relative entropy of coherence.
\end{lemma}
\begin{proof}
Recalling that the relative entropy of coherence admits the closed
expression $R_{\mathrm{c}}(\rho)=S(\overline{\rho})-S(\rho)$, it
follows that Eq.~(\ref{eq:lemma}) is equivalent to 
\begin{equation*}
S(\overline{\rho}^{AB})-S(\overline{\rho}^{B})-S(\rho^{AB})+S(\rho^{B})\leq2\log_{2}d_{A}.
\end{equation*}
The left-hand side of this inequality is a difference of two conditional
entropies: 
\begin{equation*}
S(A|B)_{\overline{\rho}}-S(A|B)_{\rho}\leq2\log_{2}d_{A}
\end{equation*}
with the states $\overline{\rho}=\overline{\rho}^{AB}$ and $\rho=\rho^{AB}$.
In the final step, we note that for any state $\rho$ the quantum
conditional entropy is bounded above and below as follows: 
\begin{equation*}
-\log_{2}d_{A}\leq S(A|B)_{\rho}\leq\log_{2}d_{A}.
\end{equation*}
This completes the proof of the Lemma.
\end{proof}

\section{Resource theory framework and our assumptions}

We will use the definition of resource theory framework given in \cite{AnshuHJ17}. Informally, resource theory consists of the set of free states  $\cF$ along with the free operations $\cG$ that map free states on some register to free states to a possibly different register. Let $\cF_{\cE}$ be the set of all operators $0 \preceq O \preceq \id$ such that the map $$\cE(\rho)\defeq O\rho O\otimes \ketbra{0} + \sqrt{\id-O^2}\rho\sqrt{\id-O^2} \otimes \ketbra{1}$$ belongs to $\cG$.

In our version of quantum state redistribution, Alice is allowed to perform arbitrary quantum operations, whereas Bob is only allowed operations from $\cG$. Our protocol will have the property that Bob will perform measurement using operators from $\cF_{\cE}$.

\begin{figure}[ht]
\centering
\begin{tikzpicture}[xscale=0.9,yscale=1.1]

\draw[ultra thick] (-3.5,6) rectangle (9,1);

\draw[thick] (-0.5,5) -- (-2,3) -- (-1.7,2) -- (1,2.5) -- (-0.5,5);
\node at (-0.2,5) {$R$};
\node at (-2.2,3) {$A$};
\node at (-2,2) {$C$};
\node at (1.2,2.5) {$B$};
\node at (-0.5,3) {$\ket{\Psi}_{RABC}$};

\node at (-0.7,5.6) {Referee};
\node at (-2.2,4.6) {Alice};
\node at (0.8,4.6) {Bob};

\draw[->, thick] (2,3.5) -- (3.5,3.5);

\draw[thick] (5.5,5) -- (4,2.5) -- (6.8,2) -- (7,3) -- (5.5,5);
\node at (5.8,5) {$R$};
\node at (3.8,2.5) {$A$};
\node at (7.2,3) {$B$};
\node at (7.1,2) {$C$};
\node at (5.8,3) {$\Phi_{RABC}$};

\node at (5.3,5.6) {Referee};
\node at (3.8,4.6) {Alice};
\node at (6.8,4.6) {Bob};

\end{tikzpicture}
\caption{\small The task of quantum state redistribution, where Alice needs to send her register $C$ to Bob, with the requirement that $\Pur(\Phi_{RABC},\ketbra{\Psi}_{RABC})\leq \eps$, for some error parameter $\eps$. Alice and Bob are allowed to have pre-shared entanglement. Bob is allowed restricted set of operations.}
 \label{fig:stateredist}
\end{figure}
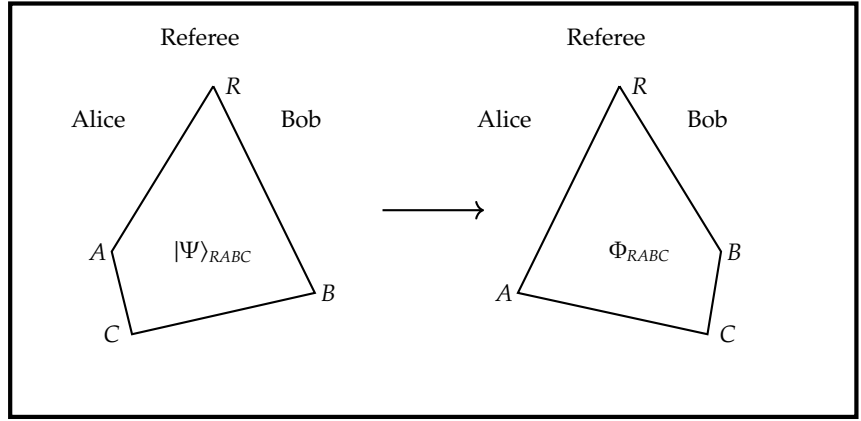

\section{An achievability bound on quantum state redistribution}
\label{sec:quantumstate}

Quantum state redistribution is the following coherent quantum task (see Figure \ref{fig:stateredist}). In this task, Alice, Bob and Referee share a pure state $\ket{\Phi}_{RABC}$, with $AC$ belonging to Alice, $B$ to Bob and $R$ to Referee. Alice needs to transfer the register $C$ to Bob, such that the final state $\Phi'_{RABC}$ satisfies $\Pur(\Phi'_{RABC},\Phi_{RABC})\leq \eps$, for a given $\eps \in (0,1)$ which is the error parameter.  Alice and Bob are allowed to have pre-shared entanglement. Furthermore, Bob can perform operations that belong to $\cG$. The communication from Alice to Bob is in the form of \textit{coherent classical bits} \cite{Harrow04}. We will refer to these as \textit{cobits}. 

\begin{definition}[Quantum state redistribution]
\label{def:qsr}
Fix an $\eps\in (0,1)$.  Consider the state $\ket{\Phi}_{RABC}$ and let Alice ($E_A$) and Bob ($E_B$) preshare an entangled state $\ket{\theta}_{E_AE_B}$ such that $\theta_{E_B}\in \cF$. An $(m ,\epsilon)$-quantum state redistribution protocol consists of 
\begin{itemize}
\item Alice's encoding isometry $\cE: \cL(ACE_A) \rightarrow \cL(AMT_A)$, 
 and
\item Bob's decoding map $\cD: \cL(MBE_B)\rightarrow \cL(BCT_B)$ such that $\cD\in \cG$.
\end{itemize}
Let the final state be
$$\Phi'_{RABCT_AT_B} \defeq \cD\circ\cE(\Phi_{RABC}\otimes \theta_{E_AE_B}).$$ 
There exists a state $\sigma_{T_AT_B}$ such that
$$\Pur(\Phi'_{RABCT_AT_B},\Phi_{RABC}\otimes \sigma_{T_AT_B})\leq \eps.$$ The number of cobits communicated is $m = \log|M|$. 
\end{definition}

Following is the main result of this section.
\begin{theorem}[Achievability bound]
\label{achievability}
Fix $\eps_1,\eps_2, \gamma\in (0,1)$ and let $\sigma_C \in \cF$ be a free quantum state. There exists an $(m, 3\eps_1+\eps_2+\gamma)$- quantum state redistribution protocol for $\ket{\Phi}_{RABC}$ for any $m$ satisfying 
\begin{eqnarray*}
m \geq \inf_{\Phi'\in \ball{\eps_1}{\Phi}}\dmax{\Phi'_{RBC}}{\Phi'_{RB}\otimes \sigma_C} - \dFeps{\Phi_{BC}}{\Phi_{B}\otimes \sigma_C}{\eps^4_2} + 2\log\left(\frac{2}{\varepsilon_1\cdot \gamma^2}\right),
\end{eqnarray*}
\end{theorem}

\begin{proof}
Let $k \defeq \inf_{\Phi'\in\ball{\eps_1}{\Phi}}\dmax{\Phi'_{RBC}}{\Phi'_{RB}\otimes \sigma_C}$, $\delta\defeq\varepsilon_1^2$ and $n \defeq \lceil\frac{2^k }{\delta}\rceil$. Let $$b\defeq \lceil \gamma^4\cdot 2^{\dFeps{\Phi_{BC}}{\Phi_{B}\otimes \sigma_C}{\eps^4_2}}\rceil$$ and $\Pi_{BC}\in \cF_{\cE}$ be the operator achieving the optimum in the definition of $\dFeps{\Phi''_{BC}}{\Phi''_{B}\otimes \sigma_C}{\eps^4_2}$. Consider the state, 
$$\mu_{RB C_1 \ldots C_n} \defeq \frac{1}{n}\sum_{j=1}^n \Phi_{RBC_j}\otimes\sigma_{C_1}\otimes\ldots\otimes\sigma_{C_{j-1}}\otimes\sigma_{C_{j+1}}\otimes\ldots \otimes\sigma_{C_n}.$$
Note that $\Phi_{RB} = \mu_{RB}$. Consider the following purification of $\mu_{RB C_1\ldots C_n}$, 
\begin{align*}
\lefteqn{\ket{\mu}_{RBJL_1\ldots L_n C_1\ldots C_n}}  \\ 
& & = \frac{1}{\sqrt{n}}\sum_{j=1}^n\ket{j}_J\ket{\tilde{\Phi}}_{RBAC_j}\otimes \ket{\sigma}_{L_1C_1}\otimes \ldots\otimes\ket{\sigma}_{L_{j-1}C_{j-1}}\otimes\ket{0}_{L_j}\otimes\ket{\sigma}_{L_{j+1}C_{j+1}}\otimes\ldots\otimes\ket{\sigma}_{L_nC_n} .
\end{align*} 
Here, $\forall j\in[n]: \ket{\sigma}_{L_jC_j}$ is a purification of  $\sigma_{C_j}$ and $\ket{\tilde{\Phi}}_{RBAC_j}$ is a purification of $\Phi_{RBC_j}$. Consider the following protocol $\cP_1$.

\begin{figure}[h]
\centering
\begin{tikzpicture}[xscale=0.9,yscale=1.2]

\node at (2.5,5) {$\frac{1}{b}$};

\draw[fill] (3.5,7) circle [radius=0.05];
\draw[fill] (3.5,8) circle [radius=0.05];
\draw[thick] (3.5,8) to [out=200,in=160] (3.5,7);
\node at (2.4,7.5) {$\Psi_{RBAC_1}$};
\node at (3.5,6.5) {$\otimes$};
\node at (2.9,6) {$\sigma_{C_2}$};
\draw[fill] (3.5,6) circle [radius=0.05];
\node at (3.5,5.5) {$\otimes$};
\draw[fill] (3.5,5) circle [radius=0.05];
\node at (3.5,4.5) {$\otimes$};
\draw[fill] (3.5,4.1) circle [radius=0.02];
\draw[fill] (3.5,3.7) circle [radius=0.02];
\draw[fill] (3.5,3.3) circle [radius=0.02];
\draw[fill] (3.5,2.9) circle [radius=0.02];
\node at (3.5,2.5) {$\otimes$};
\draw[fill] (3.5,2) circle [radius=0.05];
\node at (2.9,2) {$\sigma_{C_b}$};

\node at (4.5,5) {$+$};

\node at (5.5,5) {$\frac{1}{b}$};

\draw[fill] (6.5,7) circle [radius=0.05];
\draw[fill] (6.5,8) circle [radius=0.05];
\node at (7.1,7) {$\sigma_{C_1}$};
\draw[thick] (6.5,8) to [out=200,in=160] (6.5,6);
\node at (5.1,7) {$\Psi_{RBAC_2}$};
\node at (6.5,6.5) {$\otimes$};
\draw[fill] (6.5,6) circle [radius=0.05];
\node at (6.5,5.5) {$\otimes$};
\draw[fill] (6.5,5) circle [radius=0.05];
\node at (6.5,4.5) {$\otimes$};
\draw[fill] (6.5,4.1) circle [radius=0.02];
\draw[fill] (6.5,3.7) circle [radius=0.02];
\draw[fill] (6.5,3.3) circle [radius=0.02];
\draw[fill] (6.5,2.9) circle [radius=0.02];
\node at (6.5,2.5) {$\otimes$};
\draw[fill] (6.5,2) circle [radius=0.05];
\node at (5.9,2) {$\sigma_{C_b}$};

\draw[fill] (7.5,5.0) circle [radius=0.02];
\draw[fill] (7.9,5.0) circle [radius=0.02];
\draw[fill] (8.3,5.0) circle [radius=0.02];
\draw[fill] (8.7,5.0) circle [radius=0.02];

\node at (9.2,5) {$+$};

\node at (9.7,5) {$\frac{1}{b}$};

\draw[fill] (10.5,7) circle [radius=0.05];
\draw[fill] (10.5,8) circle [radius=0.05];
\draw[thick] (10.5,8) to [out=230,in=85] (10.1,7) to [out=265,in=95] (10.1,3) to [out=275,in=130] (10.5,2);
\node at (9.4,7.5) {$\Psi_{RBAC_b}$};
\node at (11.1,7) {$\sigma_{C_1}$};
\node at (10.5,6.5) {$\otimes$};
\node at (11.1,6) {$\sigma_{C_2}$};
\draw[fill] (10.5,6) circle [radius=0.05];
\node at (10.5,5.5) {$\otimes$};
\draw[fill] (10.5,5) circle [radius=0.05];
\node at (10.5,4.5) {$\otimes$};
\draw[fill] (10.5,4.1) circle [radius=0.02];
\draw[fill] (10.5,3.7) circle [radius=0.02];
\draw[fill] (10.5,3.3) circle [radius=0.02];
\draw[fill] (10.5,2.9) circle [radius=0.02];
\node at (10.5,2.5) {$\otimes$};
\draw[fill] (10.5,2) circle [radius=0.05];

\end{tikzpicture}
\caption{\small Bob performs restricted hypothesis testing on the state $\mu^{(2)}_{RABC_1C_2\ldots C_b}$ (depicted above). This is the state he obtains after receiving Alice's message.}
 \label{fig:convexcomb}
\end{figure}

\begin{enumerate}
\item Alice, Bob and Referee start by sharing the state $\ket{\mu}_{RBJL_1\ldots L_n C_1\ldots C_n}$ between themselves where Alice holds registers $JL_1\ldots L_n$, Referee holds the register $R$ and Bob holds the registers  $BC_1C_2\ldots C_n$.
\item Alice measures the register $J$ coherently and obtains the outcome $j \in [n]$. She sends the integer $\lfloor (j-1)/b \rfloor$ to Bob using $\lceil\log(n/b)\rceil$ cobits. 

\item Bob swaps registers $C_{b\cdot\lfloor (j-1)/b \rfloor+1},C_{b\cdot\lfloor (j-1)/b \rfloor+2},\ldots C_{b\cdot\lfloor (j-1)/b \rfloor+b}$ with the set of registers $C_1,C_2,\ldots C_b$ in that order. 

\begin{itemize}
\item At this step of the protocol, the joint state in the registers $RBAC_1C_2\ldots C_b$ is (see Figure \ref{fig:convexcomb})
$$\mu^{(2)}_{RBAC_1C_2\ldots C_b} = \frac{1}{b}\sum_{j=1}^b \ketbra{\Phi}_{RBAC_j}\otimes \sigma_{C_1}\otimes\ldots \sigma_{C_{j-1}}\otimes \sigma_{C_{j+1}}\otimes\ldots \sigma_{C_b}.$$
\end{itemize}

\item Bob performs the following quantum operation (which we label as $\cP_B$). 
\begin{enumerate}
\item Initialize $k:=1$.
\item While $k \leq b$, do: 
\begin{enumerate}
\item Measure registers $BC_k$ with the measurement $\{\Pi_{BC_k}, \id - \Pi_{BC_k}\}$. 
\item If the outcome corresponds to $\Pi_{BC_k}$, Swap $C_k, C_1$. Else set $k:= k+1$ and go to Step (i).
\end{enumerate}
\end{enumerate}
\item Final state is obtained in the registers $RABC_1$. We call it $\Phi^1_{RBAC_1}$.
\end{enumerate}   

We have the following claim. 
\begin{claim}
\label{claimprotp1}
 It holds that $\Pur(\Phi^1_{RBAC_1}, \Phi_{RBAC_1}) \leq \eps_2 + \gamma$.
\end{claim}
\begin{proof}
For brevity, set $\sigma^{(-j)}\defeq \sigma_{C_1}\otimes\ldots \sigma_{C_{j-1}}\otimes \sigma_{C_{j+1}}\otimes\ldots \sigma_{C_b}$.  Consider the projective measurement $\{\hat{\Pi}_{BCP}, \id- \hat{\Pi}_{BCP}\}$ obtained by applying Neumark's Theorem (Fact \ref{Neumark}) to the measurement $\{\Pi_{BC}, \id - \Pi_{BC}\}$. Introduce registers $P_1, \ldots P_b$ in the states $\ketbra{0}_{P_1}\otimes \ldots \ketbra{0}_{P_b}$ and let $\bar{P}:= P_1\ldots P_b$. Let $\{\Pi_{BC_kP_k}, \id - \Pi_{BC_kP_k}\}$ be the projective measurement that acts trivially on registers $P_1, \ldots P_{k-1}, P_{k+1}, \ldots P_k$.  Define $$M_k \defeq \hat{\Pi}_{BC_kP_k}(\id - \hat{\Pi}_{BC_{k-1}P_{k-1}})\ldots (\id - \hat{\Pi}_{BC_1P_1})$$ for $k\leq b$ and 
$$M_{b+1} \defeq (\id - \hat{\Pi}_{BC_{b}P_b})\ldots (\id - \hat{\Pi}_{BC_1P_1}).$$ Observe that $\{M_k\}_{k=1}^b$ form a complete set of POVM elements. From Facts \ref{noncommutativebound} and \ref{Neumark}, we have that 
\begin{eqnarray*}
&&\Pur\left(\ketbra{\Phi}_{RBAC_k}\otimes \sigma^{(-k)}\otimes \ketbra{0}^{\otimes b}_P, M_k\left(\ketbra{\Phi}_{RBAC_k}\otimes \sigma^{(-k)}\otimes \ketbra{0}^{\otimes b}_P\right)M^{\dagger}_k\right)\\ &&\leq \left((k-1)\cdot\Tr(\Pi_{BC}\Phi_B\otimes \sigma_C) + 1-\Tr(\Pi_{BC}\Phi_{BC})\right)^{1/4}\\
&& \leq \left(b\cdot 2^{-\dFeps{\Phi_{BC}}{\Phi_{B}\otimes \sigma_C}{\eps^4_2}} + \eps_2^4\right)^{1/4}\leq \eps_2 + \gamma.
\end{eqnarray*}
This implies by triangle inequality for purified distance (Fact \ref{fact:trianglepurified}) that 
\begin{eqnarray*}
&&\Pur\bigg(\frac{1}{b}\sum_k\ketbra{\Phi}_{RBAC_k}\otimes \sigma^{(-k)}\otimes \ketbra{k}\otimes \ketbra{0}^{\otimes b}_P, \\ && \frac{1}{b}\sum_k M_k\left(\ketbra{\Phi}_{RBAC_k}\otimes \sigma^{(-k)}\otimes \ketbra{0}^{\otimes b}_P\right)M^{\dagger}_k\otimes \ketbra{k}\bigg)\leq \eps_2+\gamma.
\end{eqnarray*}
From Fact \ref{Neumark}, the quantum state obtained after Bob's operation is equal to
$$\frac{1}{b}\sum_k \Tr_P\left(M_k\left(\ketbra{\Phi}_{RBAC_k}\otimes \sigma^{(-k)}\otimes \ketbra{0}^{\otimes b}_P\right)M^{\dagger}_k\right)\otimes \ketbra{k}.$$
Thus, using Fact \ref{fact:monotonequantumoperation}, Bob is able to obtain the quantum state
$$\frac{1}{b}\sum_k\ketbra{\Phi}_{RBAC_k}\otimes \sigma^{(-k)}\otimes \ketbra{k}$$
with high fidelity. Since Bob swaps register $C_k$ with $C_1$ upon obtaining the outcome $k$, which is a unitary operation and hence does not change the fidelity, the desired claim follows. 
\end{proof}

This shows that protocol $\cP_1$ succeeds with fidelity squared as given in the claim. Now we proceed to construct the actual protocol.

Consider the state,
$$\xi_{RB C_1\ldots C_n}  \defeq \Phi_{RB}\otimes \sigma_{C_1}\ldots \otimes \sigma_{C_n}. $$ 
Let $\ket{\theta}_{L_1\ldots L_n C_1\ldots C_n} = \ket{\sigma}_{L_1C_1}\otimes \ket{\sigma}_{L_2C_2}\ldots \ket{\sigma}_{L_nC_n}$ be a purification of  $\sigma_{C_1}\otimes \ldots \sigma_{C_n}$. Let 
 $$\ket{\xi}_{RABC L_1\ldots L_nC_1\ldots C_n} \defeq  \ket{\Phi}_{RABC} \otimes \ket{\theta}_{L_1\ldots L_n C_1\ldots C_n} .$$ 
Using Claim ~\ref{convexcombcor} (variant of convex split lemma) and choice of $n$ we have,
$$\F^2(\xi_{RBC_1\ldots C_n}, \mu_{RB C_1\ldots C_n})      \\ 
 \geq 1 - 9\varepsilon^2_1 .$$ 
Let $\ket{\xi'}_{RBJ L_1\ldots L_nC_1\ldots C_n}$ be a purification of $\xi_{RB C_1\ldots C_n}$ (guaranteed by Uhlmann's Theorem, Fact~\ref{uhlmann}) such that,
$$\F^2(\ketbra{\xi'}_{RBJ L_1\ldots L_nC_1\ldots C_n},\ketbra{\mu}_{RBJ L_1\ldots L_nC_1\ldots C_n}) = \F^2(\xi_{RBC_1\ldots C_n}, \mu_{RB C_1\ldots C_n})  \geq 1- 9\varepsilon^2_1.$$ 
Let $V': ACL_1\ldots L_n \rightarrow J L_1\ldots L_n$ be an isometry (guaranteed by Uhlmann's Theorem, Fact~\ref{uhlmann})  such that,
$$V'\ket{\xi}_{RABC L_1\ldots L_nC_1\ldots C_n} = \ket{\xi'}_{RBJ L_1\ldots L_nC_1\ldots C_n} .$$ 
Consider the following protocol $\cP$.
\begin{enumerate}
\item Alice, Bob and Referee start by sharing the state $\ket{\xi}_{RABC L_1\ldots L_nC_1\ldots C_n} $ between themselves where Alice holds registers $ACL_1\ldots L_n$, Referee holds the register $R$ and Bob holds the registers $BC_1\ldots C_n$. Note that  $\ket{\Psi}_{RABC}$ is provided as input to the protocol and $\ket{\theta}_{L_1\ldots L_n C_1\ldots C_n}$ is additional shared entanglement between Alice and Bob. 
\item Alice applies isometry $V'$ to obtain state $\ket{\xi'}_{RBJL_1\ldots L_nC_1\ldots C_n}$, where Alice holds registers $JL_1\ldots L_n$, Referee holds the register $R$ and Bob holds  the registers  $BC_1\ldots C_n$.
\item Alice and Bob simulate protocol $\cP_1$ from Step 2. onwards.
\end{enumerate}

\vspace{0.1in}

\noindent {\bf Error analysis:} Let $\Phi'_{RABC}$ be the output of protocol $\cP$. Since quantum maps (the entire protocol $\cP_1$ can be viewed as a quantum map from input to output) do not decrease fidelity (monotonicity of fidelity under quantum operation, Fact~\ref{fact:monotonequantumoperation}), we have,
\begin{equation*}
\F^2(\Phi^1_{RABC}, \Phi'_{RABC}) \geq  \F^2(\ketbra{\xi'}_{RBJ L_1\ldots L_nC_1\ldots C_n},\ketbra{\mu}_{RBJ L_1\ldots L_nC_1\ldots C_n}) \geq 1- 9\varepsilon^2_1.
\end{equation*}
This implies by Claim \ref{claimprotp1} and triangle inequality for purified distance (Fact \ref{fact:trianglepurified}) that $\F^2(\Phi'_{RABC},\Phi_{RABC}) \geq 1-(3\eps_1+\eps_2 +\gamma)^2$. That is, $\Phi'_{RABC} \in \ball{3\eps_1+\eps_2+\gamma}{\Phi_{RABC}}$.

\vspace{0.1in}

\noindent {\bf Communication cost:} The number of cobits communicated by Alice to Bob in $\cP$ is equal to the number of cobits communicated in $\cP_1$ and is upper bounded by:
$$\lceil\log(n/b)\rceil \leq \inf_{\Phi'\in \ball{\eps_1}{\Phi}}\dmax{\Phi'_{RBC}}{\Phi'_{RB}\otimes \sigma_C} - \dFeps{\Phi_{BC}}{\Phi_{B}\otimes \sigma_C}{\eps^4_2} + 2\log\left(\frac{2}{\varepsilon_1\cdot \gamma^2}\right).$$ 

This completes the proof.  
\end{proof}

\section{Achievability and converse for a class of resource theories}

We simplify the expression for $\dFeps{\Phi_{BC}}{\Phi_{B}\otimes \sigma_C}{\eps}$ when $\cF$ corresponds to a certain class of resource theories that contain a `collapsing map' defined below. 
\begin{definition}
\label{def:projrestheo}
\textbf{Collapsing map:} Fix a register $R$. A trace preserving map $\Delta_R: \cL(R)\rightarrow \cL(R)$ is said to be collapsing if for all operator $O \in \cL(R)$ satisfying $ 0 \preceq O \preceq \id_R$, $\Delta_R^{\dagger}(O) \in \cF_{\cE}$.  Furthermore, for registers $R,S$, we have $\Delta_{RS} = \Delta_R \otimes \Delta_S$. The map $\Delta_R$ is said to be surjective with respect to $\cF_{\cE}$ if for every $\Pi\in \cF_{\cE}$, there exists a $O\in \cL(R)$ satisfying $0 \preceq O \preceq \id_R$, such that $\Delta_R^{\dagger}(O) = \Pi$.
\end{definition}
In the resource theory of coherence, the dephasing map (in a given basis) is a natural example of a collapsing map. We prove the following lemma.
\begin{lemma}
\label{dhcoherence}
Let $\eps\in (0,1)$. Let $\cF, \cF_{\cE}$ correspond to a resource theory that contains a collapsing map (Definition \ref{def:projrestheo}). It holds that 
$$\dFeps{\Phi_{BC}}{\Phi_{B}\otimes \sigma_C}{\eps} \geq \dheps{\Delta_{BC}(\Phi_{BC})}{\Delta_B(\Phi_{B})\otimes \Delta_C(\sigma_C)}{\eps}.$$ Furthermore, if $\Delta_{BC}$ is surjective with respect to $\cF_{\cE}$, then $$\dFeps{\Phi_{BC}}{\Phi_{B}\otimes \sigma_C}{\eps} = \dheps{\Delta_{BC}(\Phi_{BC})}{\Delta_B(\Phi_{B})\otimes \Delta_C(\sigma_C)}{\eps}.$$
\end{lemma}
\begin{proof}
Let $\Pi_{BC}$ be the operator achieving the optimum in the definition of $\dheps{\Delta_{BC}(\Phi_{BC})}{\Delta_B(\Phi_{B})\otimes \Delta_C(\sigma_C)}{\eps}$. Consider $$\Tr(\Pi_{BC}\Delta_{BC}(\rho_{BC})) = \Tr(\Delta^{\dagger}_{BC}(\Pi_{BC})\rho_{BC}) \geq 1-\eps.$$ Furthermore, using the relation $\Delta_{BC} = \Delta_B \otimes \Delta_C$ as given in Definition \ref{def:projrestheo}, 
\begin{eqnarray*}
&& 2^{-\dheps{\Delta_{BC}(\Phi_{BC})}{\Delta_B(\Phi_{B})\otimes \Delta_C(\sigma_C)}{\eps}}=\Tr(\Pi_{BC}\Delta_{B}(\rho_{B})\otimes \Delta_{C}(\rho_{C})) \\ &&= \Tr(\Delta^{\dagger}_{BC}(\Pi_{BC})\rho_{B}\otimes \rho_C) \geq 2^{-\dFeps{\Phi_{BC}}{\Phi_{B}\otimes \sigma_C}{\eps}},
\end{eqnarray*}
where last inequality follows since $\Delta^{\dagger}_{BC}(\Pi_{BC}) \in \cF_{\cE}$ by assumption. This proves the first part.

The second part follows similarly. Let $\Pi_{BC}$ be the operator that achieves the optimum in the definition of $\dFeps{\Phi_{BC}}{\Phi_{B}\otimes \sigma_C}{\eps}$. Since $\Pi_{BC} \in \cF_{\cE}$, by the assumption that $\Delta_{BC}$ is surjective, there exists an operator $O_{BC}$ such that $\Delta^{\dagger}(O_{BC})= \Pi_{BC}$. Thus,
$$ \Tr(O_{BC}\Delta_{BC}(\Phi_{BC}))= \Tr(\Delta^{\dagger}(O_{BC})\Phi_{BC}) = \Tr(\Pi_{BC}\Phi_{BC}) \geq 1-\eps,$$ and using $\Delta_{BC}= \Delta_B\otimes \Delta_C$,
\begin{eqnarray*}
&&  2^{-\dheps{\Delta_{BC}(\Phi_{BC})}{\Delta_B(\Phi_{B})\otimes \Delta_C(\sigma_C)}{\eps}} \leq \Tr(O_{BC}\Delta_B(\Phi_{B})\otimes \Delta_C(\sigma_{C})) \\ && = \Tr(\Delta^{\dagger}_{BC}(O_{BC})\Phi_B\otimes \sigma_C) = \Tr(\Pi_{BC}\Phi_B\otimes \sigma_C) = 2^{-\dFeps{\Phi_{BC}}{\Phi_{B}\otimes \sigma_C}{\eps}}.
\end{eqnarray*} 
This implies that 
$$\dFeps{\Phi_{BC}}{\Phi_{B}\otimes \sigma_C}{\eps} \leq \dheps{\Delta_{BC}(\Phi_{BC})}{\Delta_B(\Phi_{B})\otimes \Delta_C(\sigma_C)}{\eps}.$$
This completes the proof.
\end{proof}

Now, we can combine Theorem \ref{achievability} (setting $\eps_3=0, \gamma = \eps_2$) and Lemma \ref{dhcoherence} to conclude the following corollary. 

\begin{corollary}[Achievability bound for resource theories with a collapsing map]
\label{coherenceachievability}
Fix $\eps_1,\eps_2\in (0,1)$ and let $\sigma_C\in \cF$ be an arbitrary quantum state. There exists an $(m, 3\eps_1+2\eps_2)$ quantum state redistribution protocol for the state $\ket{\Phi}_{RACB}$ for any $m$ that satisfies
\begin{eqnarray*}
m \geq \inf_{\Phi'\in \ball{\eps_1}{\Phi}}\dmax{\Phi'_{RBC}}{\Phi'_{RB}\otimes \sigma_C} - \dheps{\Delta_{BC}(\Phi_{BC})}{\Delta_B(\Phi_{B})\otimes \Delta_C(\sigma_C)}{\eps_2^4} + 2\log\left(\frac{2}{\varepsilon_1\cdot \eps^2_2}\right),
\end{eqnarray*}
\end{corollary}

\subsection{An achievability result in the asymptotic and i.i.d. setting}

We show an achievability result assuming the following properties in the resource theory.
\begin{itemize}
\item \textbf{P. 1}: Fix a register $C$. For every free state $\sigma_C\in \cF$, it holds that $\Delta_C(\sigma_C) = \sigma_C$.
\item \textbf{P. 2}: Fix a register $C$. For quantum states $\rho_C, \rho'_C$, it holds that $$\Tr(\rho_C\log \Delta_C(\rho'_C)) = \Tr(\Delta_C(\rho_C)\log \Delta_C(\rho'_C)).$$
\end{itemize}
The following theorem holds.
\begin{theorem}
\label{theo:tradeoff}
Fix $\eps, \delta \in (0,1)$ and $\sigma_C \in \cF$. There exists a large enough $n$ such that there exists a $(n(Q_C+ \delta), \eps)$ quantum state redistribution protocol for the quantum state $\ketbra{\Phi}^{\otimes n}_{RABC}$, if 
\begin{eqnarray*}
Q_C &\geq& \relent{\Phi_{RBC}}{\Phi_{RB}\otimes \sigma_C} - \relent{\Delta_{BC}(\Phi_{BC})}{\Delta_B(\Phi_{B})\otimes \Delta_C(\sigma_C)} \\
& = & \condmutinf{R}{C}{B}_{\Phi} + \relent{\Phi_{BC}}{\Delta_{BC}(\Phi_{BC})} - \relent{\Phi_B}{\Delta_{B}(\Phi_{B})} .
\end{eqnarray*}
\end{theorem}
\begin{proof}
We first show the inequality. Using Corollary \ref{coherenceachievability}, there exists a $(m, \eps)$ quantum state redistribution protocol for the state $\ket{\Phi}^{\otimes n}_{RACB}$ for any $m$ that satisfies
\begin{eqnarray*}
m &\geq& \inf_{\Phi'\in \ball{\eps/5}{\Phi^{\otimes n}}}\dmax{\Phi'_{R^nB^nC^n}}{\Phi'_{R^nB^n}\otimes \sigma_{C^n}}\\ &-& \dheps{\Delta_{BC}(\Phi_{BC})^{\otimes n}}{\Delta_B(\Phi_{B})^{\otimes n}\otimes \Delta_C(\sigma_C)^{\otimes n}}{\eps^4/5^4} + 2\log\left(\frac{250}{\eps^3}\right),
\end{eqnarray*} 
where we have used the property that $\Delta_{B^nC^n} = \Delta_{BC}^{\otimes n}$. From \cite[Lemma 3]{AnshuJW17},
$$\inf_{\Phi'\in \ball{\eps/5}{\Phi^{\otimes n}}}\dmax{\Phi'_{R^nB^nC^n}}{\Phi'_{R^nB^n}\otimes \sigma_{C^n}} \leq \dmaxeps{\Phi^{\otimes n}_{RBC}}{\Phi^{\otimes n}_{RB}\otimes \sigma^{\otimes n}_{C}}{\eps/10} + \log\frac{300}{\eps^2}.$$ Thus, it suffices to have
\begin{eqnarray*}
m &\geq&\dmaxeps{\Phi^{\otimes n}_{RBC}}{\Phi^{\otimes n}_{RB}\otimes \sigma^{\otimes n}_{C}}{\eps/10} \\ &-& \dheps{\Delta_{BC}(\Phi_{BC})^{\otimes n}}{\Delta_B(\Phi_{B})^{\otimes n}\otimes \Delta_C(\sigma_C)^{\otimes n}}{\eps^4/5^4} + 8\log\left(\frac{10}{\eps}\right).
\end{eqnarray*} 
Using Facts \ref{dmaxequi} and \ref{gaussianupper}, we conclude that it suffices to have
\begin{eqnarray*}
m &\geq& n\left(\relent{\Phi_{RBC}}{\Phi_{RB}\otimes \sigma_C} - \relent{\Delta_{BC}(\Phi_{BC})}{\Delta_B(\Phi_{B})\otimes \Delta_C(\sigma_C)} + \mathcal{O}\left(\sqrt{\frac{\log\frac{1}{\eps}}{n}}\right)\right).
\end{eqnarray*} 
Letting $n$ large enough such that $\delta \geq \mathcal{O}\left(\sqrt{\frac{\log\frac{1}{\eps}}{n}}\right)$, the inequality follows.

For the equality, consider,
\begin{align*}
 \lefteqn{\relent{\Phi_{RBC}}{\Phi_{RB}\otimes \sigma_C} - \relent{\Delta_{BC}(\Phi_{BC})}{\Delta_B(\Phi_{B})\otimes \Delta_C(\sigma_C)}} \\
  & = \relent{\Phi_{RBC}}{\Phi_{RB}\otimes \sigma_C} - \relent{\Phi_{BC}}{\Phi_{B}\otimes \sigma_C} + \relent{\Phi_{BC}}{\Phi_{B}\otimes \sigma_C} - \relent{\Delta_{BC}(\Phi_{BC})}{\Delta_B(\Phi_{B})\otimes \Delta_C(\sigma_C)} \\
  & = \relent{\Phi_{RBC}}{\Phi_{RB}\otimes \Phi_C} - \relent{\Phi_{BC}}{\Phi_{B}\otimes \Phi_C} + \relent{\Phi_{BC}}{\Phi_{B}\otimes \sigma_C} - \relent{\Delta_{BC}(\Phi_{BC})}{\Delta_B(\Phi_{B})\otimes \Delta_C(\sigma_C)} \\
  & = \condmutinf{R}{C}{B}_{\Phi} + \relent{\Phi_{BC}}{\Phi_{B}\otimes \sigma_C} - \relent{\Delta_{BC}(\Phi_{BC})}{\Delta_B(\Phi_{B})\otimes \Delta_C(\sigma_C)} \\
  & =  \condmutinf{R}{C}{B}_{\Phi} + \relent{\Phi_{BC}}{\Delta_{BC}(\Phi_{BC})} - \relent{\Phi_B}{\Delta_{B}(\Phi_{B})} \quad \mbox{(using property P. 2)}\\ & \hspace{0.1in}- \Tr \Phi_C \log \sigma_C +  \Tr \Delta_C(\Phi_C) \log \Delta_C(\sigma_C) \\ 
  & =  \condmutinf{R}{C}{B}_{\Phi} + \relent{\Phi_{BC}}{\Delta_{BC}(\Phi_{BC})} - \relent{\Phi_B}{\Delta_{B}(\Phi_{B})} . \quad \quad \mbox{(since $\sigma_C = \Delta(\sigma_C)$, using property P. 1)} 
\end{align*}

\end{proof}

\subsection{A converse bound in the asymptotic and i.i.d. setting}

We show some converse results for the task of quantum state redistribution. We assume the following properties for the strategy followed by Alice and Bob, which are satisfied in our achievabiliy result in Theorem \ref{theo:tradeoff}. 
\begin{itemize}
\item \textbf{P. 3}: Alice communicates classical messages to Bob, with the classical basis chosen such that the message $\omega_M$ satisfies $\Delta_M(\omega_M) = \omega_M$. 
\item \textbf{P. 4}: Bob's decoding map $\cD$ (see Definition \ref{def:qsr}) is such that for every quantum state $\tau_{MBE_B}$, there exists a quantum state $\tau'_{BCT_B}$ such that $\cD(\Delta_{MBE_B}(\tau_{MBE_B})) = \Delta_{BCT_B}(\tau'_{BCT_B})$
\end{itemize}
Property P. 3 is without loss of generality up to a factor of $2$, as for a quantum message $M$ sent from Alice to Bob, they can apply the port-based teleportation scheme \cite{Portbased08, Portbased09} to communicate classical message of $2\log|M| + \log\frac{1}{\delta}$with a small increase in error by $\delta$. Property P. 4 says that Bob's decoding operation takes `diagonal states' to 'diagonal states' in the basis defined by $\Delta$.  Since $\Delta$ is the dephasing map in the resource theory of coherence, the decoding operation takes diagonal states to diagonal states. Thus, all free operations are included in Bob's set of decoding operations that satisfy P. 4.

We define following quantity (that is necessary for the application of Fact \ref{relentresourcecont}).
$$C_R(\cF) \defeq \max_{n >0}\frac{1}{n}\min_{\tau_{R^n}}\|\log \Delta_{R^n}(\tau_{R^n})\|_{\infty}.$$
It is finite for the resource theory of coherence, as 
$$\frac{1}{n}\min_{\tau_{R^n}}\|\log \Delta_{R^n}(\tau_{R^n})\|_{\infty} = \frac{1}{n}\|\log \frac{\id_{R^n}}{|R|^n}\|_{\infty} = \log|R|.$$ We have the following result.
\begin{theorem}[Converse]
\label{theo:converse}
 Fix $\eps \in(0,1)$. For each $n \geq 1$, let there be a $(Q_C(n,\eps), \eps)$ quantum state redistribution protocol for the quantum state $\ketbra{\Phi}^{\otimes n}_{RABC}$.  It holds that 
\begin{eqnarray*}
\lim_{\eps\rightarrow 0}\lim_{n\rightarrow \infty} \frac{Q_C(n,\eps)}{n}  &\geq& \condmutinf{R}{C}{B}_{\Phi} + \relent{\Phi_{BC}}{\Delta_{BC}(\Phi_{BC})} -  \relent{\Phi_B}{\Delta_B(\Phi_B)}.
\end{eqnarray*}
\end{theorem}
\begin{proof}
Fix $n \geq 1$. Consider a $(Q_C(n,\eps), \eps)$ quantum state redistribution protocol for the quantum state $\ketbra{\Phi}^{\otimes n}_{RABC}$. Let $E_B$ be the register holding Bob's entanglement $\theta_{E_B} \in \cF$ and $M$ be the register holding Alice's message. Let $\omega_{R^nB^nME_B}$ be the quantum state on the registers of Bob and Reference after Alice's message. It holds that $\omega_{R^nB^nE_B} = \Phi^{\otimes n}_{RB}\otimes \theta_{E_B}$.  Consider
\begin{eqnarray*}
&&\relent{\omega_{R^nB^nME_B}}{\Phi^{\otimes n}_R\otimes (\Delta_B(\Phi_B))^{\otimes n}\otimes \omega_M\otimes \theta_{E_B}}\\ 
&& = \relent{\omega_{R^nB^nME_B}}{\Phi^{\otimes n}_R\otimes \Phi^{\otimes n}_B\otimes \theta_{E_B}\otimes \omega_M} + \relent{\Phi^{\otimes n}_B}{(\Delta_B(\Phi_B))^{\otimes n}} \\
&& = \relent{\omega_{R^nB^nME_B}}{\Phi^{\otimes n}_{RB}\otimes \theta_{E_B}\otimes \omega_M} + \relent{\Phi^{\otimes n}_B}{(\Delta_B(\Phi_B))^{\otimes n}} + n\mutinf{R}{B}_{\Phi}\\
&& = \relent{\omega_{R^nB^nME_B}}{\omega_{R^nB^nE_B}\otimes \omega_M} + n \relent{\Phi_B}{\Delta_B(\Phi_B)} + n\mutinf{R}{B}_{\Phi}\\
&& = \mutinf{R^nB^nE_B}{M}_{\omega} + n \relent{\Phi_B}{\Delta_B(\Phi_B)} + n\mutinf{R}{B}_{\Phi}\\
&& \leq \log|M| + n \relent{\Phi_B}{\Delta_B(\Phi_B)} + n\mutinf{R}{B}_{\Phi}.
\end{eqnarray*} 
The last inequality holds since Alice's sends cobits, as a result of which the quantum state in registers $R^nB^nE_BM$ is classical-quantum with $M$ beng classical. Let Bob apply a quantum map $\cK: \cL(B^nME_B)\rightarrow \cL(B^nC^n)$, such that $\cK\in \cG$ and $(\id_{R^n} \otimes \cK)(\omega_{R^nB^nME_B}) = \Phi'_{R^nB^nC^n}$, where $\Phi'_{R^nB^nC^n} \in \ball{\eps}{\Phi^{\otimes n}_{RBC}}$ and $\Phi'_{R^n}=\Phi^{\otimes n}_R$. Since $\omega_M = \Delta_M(\omega_M), \theta_{E_B} = \Delta_{E_B}(\theta_{E_B}) $ (Properties P. 1, P. 3) we have,
\begin{eqnarray*}
&&\relent{\omega_{R^nB^nME_B}}{\Phi^{\otimes n}_R\otimes (\Delta_B(\Phi_B))^{\otimes n}\otimes \omega_M\otimes \theta_{E_B}}\\ 
&& \geq \relent{\Phi'_{R^nB^nC^nT_B}}{\Phi^{\otimes n}_R\otimes \cK\left((\Delta_B(\Phi_B))^{\otimes n}\otimes \omega_M\otimes \theta_{E_B}\right)} \\ 
&& = \relent{\Phi'_{R^nB^nC^n}}{\Phi^{\otimes n}_R\otimes \Phi'_{B^nC^n}} + \relent{\Phi'_{B^nC^n}}{\cK\left((\Delta_B(\Phi_B))^{\otimes n}\otimes \omega_M\otimes \theta_{E_B}\right)} \\
&& \geq \relent{\Phi'_{R^nB^nC^n}}{\Phi^{\otimes n}_R\otimes \Phi'_{B^nC^n}} + \inf_{\sigma_{B^nC^n}}\relent{\Phi'_{B^nC^n}}{\Delta_{B^nC^n}(\sigma_{B^nC^n})} \quad \mbox{(Property P. 4)}\\
&& = \mutinf{R^n}{B^nC^n}_{\Phi'} + \inf_{\sigma_{B^nC^n}}\relent{\Phi'_{B^nC^n}}{\Delta_{B^nC^n}(\sigma_{B^nC^n})} \quad (\text{as }\Phi'_{R^n}=\Phi^{\otimes n}_R)\\
&& \geq n\mutinf{R}{BC}_{\Phi} + \inf_{\sigma_{B^nC^n}}\relent{\Phi_{B^nC^n}}{\Delta_{B^nC^n}(\sigma_{B^nC^n})} - (5 + C_{BC}(\cF))n\eps\log|RBC| \quad (\text{Facts \ref{fact:fannes} and \ref{relentresourcecont}}) \\
&& \geq n\mutinf{R}{BC}_{\Phi} + n \relent{\Phi_{BC}}{\Delta_{BC}(\Phi_{BC})} - (5 + C_{BC}(\cF))n\eps\log|RBC| .
\end{eqnarray*} 
Combining, we obtain,
$$\log|M| \geq n\condmutinf{R}{C}{B}_{\Phi} + n \relent{\Phi_{BC}}{\Delta_{BC}(\Phi_{BC})} -  n \relent{\Phi_B}{\Delta_B(\Phi_B)} - (5 + C_{BC}(\cF))n\eps\log|RBC|.$$
Dividing by $n$ and taking $\eps\rightarrow 0$, the proof concludes.
\end{proof}

\subsection{Implication for quantum communication in the resource theory of coherence}

It can be verified that the resource theory of coherence satisfies the properties P.1 to P.4 mentioned earlier with $\Delta$ being the dephasing map. Thus, we can apply Theorems \ref{theo:tradeoff} and \ref{theo:converse} along with the fact that superdense coding  is achievable via free operations in coherence theory, to conclude the following. 

\begin{theorem}
\label{theo:coherencerates}
Fix $\eps \in (0,1)$. There exists a large enough $n$ such that there exists a quantum state redistribution protocol for the quantum state $\ketbra{\Phi}^{\otimes n}_{RABC}$ with quantum communication cost $Q(n, \eps)$ and error $\eps$ (in purified distance), if 
\begin{eqnarray*}
\lim_{n\rightarrow \infty}\frac{1}{n}Q(n,\eps) \geq \frac{1}{2}\left(\condmutinf{R}{C}{B}_{\Phi} + \relent{\Phi_{BC}}{\Delta_{BC}(\Phi_{BC})} - \relent{\Phi_B}{\Delta_{B}(\Phi_{B})}\right).
\end{eqnarray*}
Furthermore, let $Q(n,\eps)$ be the quantum communication cost of any protocol achieving the quantum state redistribution of $\ketbra{\Phi}^{\otimes n}_{RABC}$ with error $\eps$ (in purified distance). Then it holds that
\begin{eqnarray*}
\lim_{\eps\rightarrow 0}\lim_{n\rightarrow \infty}\frac{1}{n}Q(n,\eps) \geq \frac{1}{2}\left(\condmutinf{R}{C}{B}_{\Phi} + \relent{\Phi_{BC}}{\Delta_{BC}(\Phi_{BC})} - \relent{\Phi_B}{\Delta_{B}(\Phi_{B})}\right).
\end{eqnarray*}
\end{theorem}

\section{Quantum state splitting in a large class of resource theories}

Quantum state splitting is a subtask of quantum state redistribution (Definition \ref{def:qsr}), where the register $B$ is absent. Formally, it is defined as follows.

\begin{definition}[Quantum state splitting]
\label{def:qss}
Fix an $\eps\in (0,1)$.  Consider the state $\ket{\Phi}_{RAC}$ and let Alice ($E_A$) and Bob ($E_B$) preshare an entangled state $\ket{\theta}_{E_AE_B}$ such that $\theta_{E_B}\in \cF$. An $(m ,\epsilon)$-quantum state redistribution protocol consists of 
\begin{itemize}
\item Alice's encoding isometry $\cE: \cL(ACE_A) \rightarrow \cL(AMT_A)$, 
 and
\item Bob's decoding map $\cD: \cL(ME_B)\rightarrow \cL(CT_B)$ such that $\cD\in \cG$.
\end{itemize}
Let the final state be
$$\Phi'_{RACT_AT_B} \defeq \cD\circ\cE(\Phi_{RAC}\otimes \theta_{E_AE_B}).$$ 
There exists a state $\sigma_{T_AT_B}$ such that
$$\Pur(\Phi'_{RACT_AT_B},\Phi_{RAC}\otimes \sigma_{T_AT_B})\leq \eps.$$ The number of cobits communicated is $m = \log|M|$. 
\end{definition}

Similar to Theorem \ref{theo:tradeoff}, we can show the following theorem for a large class of resource theories that allow swap operation as a free operation. Some examples include the resource theories of coherence, therodynamics, purity, non-uniformity and asymmetry, as discussed in \cite{AnshuHJ17}. We note that our result now holds even for resource theories that \textit{do not} contain a collapsing map (Definition \ref{def:projrestheo}), as Bob's operation simply involves performing a swap operation.

\begin{theorem}
\label{theo:statesplit}
Fix $\eps, \delta \in (0,1)$ and $\sigma_C \in \cF$. There exists a large enough $n$ such that there exists a $(n(Q_C+ \delta), \eps)$ quantum state splitting protocol for the quantum state $\ketbra{\Phi}^{\otimes n}_{RAC}$, if 
\begin{eqnarray*}
Q_C &\geq& \lim_{n\rightarrow \infty}\frac{1}{n}\inf_{\sigma_{C^n}}\relent{\Phi^{\otimes n}_{RC}}{\Phi^{\otimes n}_{R}\otimes \sigma_{C^n}} \\
& = & \mutinf{R}{C}_{\Phi} + \lim_{n\rightarrow \infty}\frac{1}{n} \relent{\Phi^{\otimes n}_{C}}{\sigma_{C^n}}.
\end{eqnarray*}
\end{theorem}

A matching converse is shown below, assuming that Alice's message belongs to $\cF$, which closely follows the proof of Theorem \ref{theo:converse}.
\begin{theorem}
\label{theo:splitconverse}
 Fix $\eps \in(0,1)$. For each $n \geq 1$, let there be a $(Q_C(n,\eps), \eps)$ quantum state splitting protocol for the quantum state $\ketbra{\Phi}^{\otimes n}_{RAC}$.  It holds that 
\begin{eqnarray*}
\lim_{\eps\rightarrow 0}\lim_{n\rightarrow \infty} \frac{Q_C(n,\eps)}{n}  &\geq& \mutinf{R}{C}_{\Phi} + \lim_{n\rightarrow \infty}\frac{1}{n}\inf_{\sigma_{C^n}\in \cF}\relent{\Phi^{\otimes n}_{C}}{\sigma_{C^n}}.
\end{eqnarray*}
\end{theorem}
\begin{proof}
Fix $n \geq 1$. Consider a $(Q_C(n,\eps), \eps)$ quantum state splitting protocol for the quantum state $\ketbra{\Phi}^{\otimes n}_{RAC}$. Let $E_B$ be the register holding Bob's entanglement $\theta_{E_B} \in \cF$ and $M$ be the register holding Alice's message. Let $\omega_{R^nME_B}$ be the quantum state on the registers of Bob and Reference after Alice's message. It holds that $\omega_{R^nE_B} = \Phi^{\otimes n}_{R}\otimes \theta_{E_B}$. Let Bob apply a quantum map $\cK: \cL(ME_B)\rightarrow \cL(C^n)$, such that $\cK\in \cG$ and $(\id_{R^n} \otimes \cK)(\omega_{R^nME_B}) = \Phi'_{R^nC^n}$, where $\Phi'_{R^nC^n} \in \ball{\eps}{\Phi^{\otimes n}_{RC}}$ and $\Phi'_{R^n}=\Phi^{\otimes n}_R$. Since $\omega_M \in \cF, \theta_{E_B} \in \cF$, we have,
\begin{eqnarray*}
&& \log|M| \geq \relent{\omega_{R^nME_B}}{\Phi^{\otimes n}_R\otimes  \omega_M\otimes \theta_{E_B}}\\ 
&& \geq \relent{\Phi'_{R^nC^nT_B}}{\Phi^{\otimes n}_R\otimes \cK\left(\omega_M\otimes \theta_{E_B}\right)} \\ 
&& = \relent{\Phi'_{R^nC^n}}{\Phi^{\otimes n}_R\otimes \Phi'_{C^n}} + \relent{\Phi'_{C^n}}{\cK\left(\omega_M\otimes \theta_{E_B}\right)} \\
&& \geq \relent{\Phi'_{R^nC^n}}{\Phi^{\otimes n}_R\otimes \Phi'_{C^n}} + \inf_{\sigma_{C^n}\in \cF}\relent{\Phi'_{C^n}}{\sigma_{B^nC^n}}\\
&& = \mutinf{R^n}{C^n}_{\Phi'} + \inf_{\sigma_{C^n}\in \cF}\relent{\Phi'_{C^n}}{\sigma_{C^n}} \quad (\text{as }\Phi'_{R^n}=\Phi^{\otimes n}_R)\\
&& \geq n\mutinf{R}{C}_{\Phi} + \inf_{\sigma_{C^n}\in \cF}\relent{\Phi_{C^n}}{\sigma_{C^n}} - (5 + C_{BC}(\cF))n\eps\log|RBC| \quad (\text{Facts \ref{fact:fannes} and \ref{relentresourcecont}}) 
\end{eqnarray*} 
Combining, we obtain,
$$\log|M| \geq n\mutinf{R}{C}_{\Phi} + \inf_{\sigma_{C^n}\in \cF}\relent{\Phi_{C^n}}{\sigma_{C^n}} - (5 + C_{BC}(\cF))n\eps\log|RBC|.$$
Dividing by $n$ and taking $\eps\rightarrow 0$, the proof concludes.
\end{proof}

\end{document}